\documentclass[12pt]{article}
\usepackage[T2A]{fontenc}
\usepackage[utf8]{inputenc}
\usepackage[english]{babel}
\usepackage{amsmath,amsfonts,amssymb,mathrsfs,amscd,amsthm}
\usepackage{hyperref}
\hypersetup{colorlinks=true,linkcolor=red,citecolor=blue,urlcolor=blue}
\usepackage{appendix}
\usepackage[matrix,arrow,curve]{xy}
\usepackage[a4paper, margin = 1in]{geometry}
\usepackage{csquotes}

\usepackage[giveninits=true, maxnames=99, autolang=other, sorting=none, doi=false,isbn=false,url=false,eprint=false]{biblatex}
\renewbibmacro{in:}{}
\newtheorem{theorem}{Theorem}
\newtheorem{proposition}{Proposition}
\newtheorem{remark}{Remark}
\newtheorem{definition}{Definition}

\def\M{\mathcal{M}}     %operator
\def\I{\mathbb{I}}      %identity
\DeclareMathOperator{\Tr}{Tr}
\DeclareMathOperator{\re}{Re}

\DeclareMathOperator{\Span}{Span}
\DeclareMathOperator{\Ker}{Ker}

\DeclareMathOperator{\Stab}{Stab}

\newcommand{\CC}{\mathbb{C}}
\newcommand{\RR}{\mathbb{R}}
\DeclareMathOperator{\orb}{orb}
\DeclareMathOperator{\diag}{diag}

\title{\bf Quantum channels, complex Stiefel manifolds, and optimization}

\author{Ivan Russkikh$^{1,2}$\footnote{E-mail: \url{ivanrus2002@gmail.com}}, Boris Volkov$^{1,2,3}$\footnote{E-mail: \url{borisvolkov1986@gmail.com}}, Alexander Pechen$^{1,3,4}$\footnote{E-mail: \url{apechen@gmail.com}}}

\date{}

\begin{document}
\maketitle
\noindent $^1$ Department of Mathematical Methods for Quantum Technologies, Steklov Mathematical Institute of Russian Academy of Sciences, 8 Gubkina Str., Moscow, 119991, Russia\\
$^2$ Moscow Institute of Physics and Technology, Institutskiy per. 9, Dolgoprudny, Moscow Region 141701, Russia\\
$^3$ University of Science and Technology MISIS, 4 Leninsky Prosp., Moscow 119049, Russia\\
$^4$ Ivannikov Institute for System Programming of the Russian Academy of Sciences, Alexandra Solzhenitsyna str. 25, Moscow 109004, Russia

\begin{abstract}
Most general dynamics of an open quantum system is commonly represented by a quantum channel, which is a completely positive trace-preserving map (CPTP or Kraus map). Well-known representations of quantum channels are described by Choi matrices and by Kraus operator-sum representation (OSR). As was shown before, one can use Kraus OSR to parameterize quantum channels by points of a suitable quotient of some complex Stiefel manifold by the action of the unitary group. In this work, we establish a homeomorphism between the topological space of quantum channels and the quotient of the complex Stiefel manifold. This homeomorphism can be applied to various quantum optimization problems. As an example, we apply it to the analysis of extrema points for a wide variety of quantum control objective functionals defined on the complex Stiefel manifolds, including mean value, generation of quantum gates, thermodynamic quantities involving entropy, etc. Finally, a metric on the space of quantum channels induced by the Riemannian metric on the Stiefel manifold is defined, and we show that it is a generalization of the Bures angle between density matrices.
\end{abstract}

\vspace{0.2cm}


% I. Russkikh ORCID 0009-0004-0204-4442
% B. Volkov ORCID 0000-0002-6430-9125
% A. Pechen ORCID 0000-0001-8290-8300

\section{Introduction}

Quantum channels, i.e. completely positive trace-preserving (CPTP) linear maps, also called Kraus maps, represent the most general dynamics of quantum systems~\cite{Kraus1983, Breuer_Book_1993, Nielsen_Chuang_2010, Wilde_Book_2017, Holevo_Book_2019} non-CP maps are also sometimes considered~\cite{Pechukas_1994, SHAJI200548} and the problem of determining the most general conditions which lead to a CP evolution is studied, e.g.~\cite{Shabani_Lidar_2009}). Problems involving optimization of various objective functions over quantum channels appear in many areas of quantum physics, including quantum control~\cite{KochEPJQuantumTechnol2022, BrifNewJPhys2010, GoughPhilTransRSocA2012, KochJPhysCondensMatter2016}, gate generation for quantum computation~\cite{West_Lidar_Fong_Gyure_2010, SchulteHerbruggenSporlKhanejaGlaser2011, goerz2014}, quantum information~\cite{Hayden_Jozsa_Petz_Winter_2004, Yang_Horodecki_Wang_2008,Amosov2024}, quantum state discrimination~\cite{Kim_Kwon_2023}, quantum communication~\cite{Jones_Kirby_Brodsky_2018}, generalized quantum transport~\cite{Duvenhage_2021}, etc. Common parameterizations of quantum channels use Choi matrices~\cite{choi75} or Kraus operator-sum representation (OSR)~\cite{Kraus1983}. In~\cite{PechenJPA2008, OzaJPA2009} it was shown that Kraus maps, which fundamentally model the open quantum system dynamics, can be parameterized by points of special quotients of some complex Stiefel manifolds, which are particular manifolds with orthogonality constraints. More precisely, given a quantum channel $\Phi$ for an $n$-level quantum system, its action on any density matrix $\rho$ can be written using Kraus OSR as $\Phi(\rho)=\sum_{i=1}^{n^2} K_i\rho K^\dagger_i$ (as required by complete positivity), where $\sum_{i=1}^{n^2} K_i^\dagger K_i=\mathbb I$ (as required by trace preservation). That allows to construct the matrix $K=(K_1;\dots ; K_{n^2})$ as the column composed of the matrices $K_i$. The matrix $K$ should satisfy the constraint $K^\dagger K=\mathbb I$. All matrices satisfying this constraint form the complex Stiefel manifold $V_n(\mathbb{C}^{n^3})$. The construction is generalized in this work to the case when input and output spaces have different dimensions $n$ and $m$. Since the same quantum channel can be represented using different Kraus OSRs, this non-uniqueness determines an equivalence of certain points of the complex Stiefel manifold $V_n(\mathbb{C}^{nm^2})$, precisely an equivalence under the action of the unitary group $\mathrm{U}(nm)$, which determines the non-uniqueness of the Kraus OSR, and for the one-to-one parametrization of quantum channels one has to consider a proper quotient of the complex Stiefel manifold under the action of this unitary group.   

The theory of optimization over manifolds with orthogonal constraints, including Stiefel manifolds, was developed in the fundamental work~\cite{Edelman1998}. However, the case of complex Stiefel manifolds was not considered there. In Sec.~2.1, the authors explain that ``For simplicity of exposition, but for no fundamental reason, we will concentrate on real matrices. All ideas carry over naturally to complex matrices.'' And later in Sec. 4.6 they conclude that ``In full generality the $X$ are complex, but the real case applies for physical systems of large extent that we envisage for this application, and we, accordingly, take $X$ to be real in this discussion.'' The well-known book~\cite{Absil2008} also treats only the case of real Stiefel manifolds. 

However, as was shown in~\cite{PechenJPA2008, OzaJPA2009}, complex Stiefel manifolds appear naturally in a wide field of physical phenomena described by quantum mechanics, namely in the description of open quantum systems, where complex Stiefel manifolds (strictly speaking their special quotients) can be used to parameterize quantum channels. This motivates the practical importance of the development of optimization methods and algorithms over complex Stiefel manifolds. The corresponding theory was developed in~\cite{PechenJPA2008, OzaJPA2009}, where in particular explicit analytical expressions for the gradient and Hessian of quantum control objectives representing the mean value of a quantum observable were derived, the quotient was considered, and the Riemannian and constrained optimizations were studied. 

At the same time, a relation between parametrizations of quantum channels by points of the quotient of a complex Stiefel manifold and the common parametrization of quantum channels by Choi matrices has not been established. This work aims to study such a relation. First, we rigorously establish for a general $n$-level quantum system a homeomorphism between the topological space of quantum channels and points of the quotient of the complex Stiefel manifold under the action of the unitary group (Theorem~\ref{homeomorhic}). We also show that the orbits of the unitary group action on the Stiefel manifold are diffeormorphic to the complex Stiefel manifolds $V_k(\mathbb{C}^{nm})$, where $k$ is the Kraus rank of the corresponding quantum channel (Theorem~\ref{Stiefel_repr}). 

Then, as an application of this homeomorphism, we study kinematic quantum control landscapes for various optimization problems involving quantum channels for open quantum systems. The established homeomorphism between Choi matrices and complex Stiefel manifolds allows to study local extrema of important quantum control objective functionals over Stiefel manifolds. Using it, we prove the absence of traps (i.e.\ local but not global extrema of the objective function) for a wide variety of quantum control phenomena. Previously, the absence of kinematic traps for optimization over quantum channels was proven only for the problem of maximizing the mean value of a quantum observable. For single-qubit systems, this was done by parameterising the set of quantum channels by points of the Stiefel manifold~\cite{PechenJPA2008} and using the Lagrange multiplier method. For an arbitrary dimension, the absence of kinematic traps was proven by lifting the quantum channel to the unitary dynamics in a larger space~\cite{WuPechen2008} via Stinespring dilation. Here we prove the absence of traps for a much broader class of functionals: not only for the functional describing the mean value of a quantum observable but also for various functionals describing generation of a quantum channel and more generally, any convex objective functionals for $n$-level open quantum systems for an arbitrary dimension $n$. It is important that since the used parametrizations describe in general both Markovian and non-Markovian cases, the approach based on complex Stiefel manifolds applies to both these cases as well. Note that methods of convex analysis were previously used to study quantum control objective functionals defined on the convex set of density matrices, e.g. in~\cite{PechenRabitz2010}. Finally, using the Riemannian metric on the complex Stiefel manifold, we introduce a metric on the set of quantum channels and show that it is a generalization of the Bures angle between density matrices.
 
The paper is organized as follows. Sec.~\ref{Quantum_channels} contains general information on quantum channels and their representation by Choi matrices. In Sec.~\ref{Kraus_representation_and_Stiefel_manifold}, we provide preliminaries on Kraus OSR and complex Stiefel manifolds and present our main results on the homeomorphism between the space of quantum channels and the quotient space of the Stiefel manifold with respect to the action of the unitary group, as well as explicitly describe the orbits of this action.  In Sec.~\ref{Optimization}, formulations of various optimization problems over quantum channels are given. In Sec.~\ref{Absence_of kinematic_traps}, applications of the obtained homeomorphism are provided for the analysis of kinematic quantum control landscapes and proving the absence of traps on the complex Stiefel manifolds for a broad class of objective functionals including the mean value of a quantum observable (proved before in~\cite{PechenJPA2008, WuPechen2008} by different methods), various functionals describing the generation of a quantum channel, and, more generally, any convex objective functionals for $n$-level open quantum systems of arbitrary dimension $n$. In Sec.~\ref{sec:metric}, the metric on the set of quantum channels which are induced by the Riemannian metric on the complex Stiefel manifold are introduced. We prove that this metric coincides with the Bures angle for the case of density matrices. Proofs of Theorems~\ref{homeomorhic} and~\ref{Stiefel_repr} are given in Appendices~\ref{A} and~\ref{B}, respectively. The basic necessary information from convex analysis is summarized in Appendix~\ref{C}. Conclusions and Discussion Sec.~\ref{Conclusions} summarizes this work.

\section{Quantum channels and Choi matrices}
\label{Quantum_channels}

In this section, we provide necessary preliminaries about quantum channels and Choi matrices. More information is available, e.g.\ in~\cite{Holevo_Book_2019,watrous2018theory}.

By $\|\mathord{\cdot}\|_F$ we denote the Frobenius (Hilbert-Schmidt) norm on the space $\CC^{n\times m}$ of complex $n\times m$ matrices, i.e.\ $\|A\|_F = \sqrt{\Tr(A^{\dagger}A)}$ for $A\in \CC^{n\times m}$. Denote the space of $n\times n$ complex matrices as $\mathcal{M}_n=\mathbb{C}^{n\times n}$ and the space of linear maps from $\mathcal{M}_n$ to $\mathcal{M}_m$ as $\mathcal{L}(\mathcal{M}_n, \mathcal{M}_m)$. Since $\mathcal{L}(\mathcal{M}_n, \mathcal{M}_m)\cong \CC^{n^2\times m^2}$, we assume that this space is endowed with the Frobenius norm. By $A\geq 0$ we denote the fact that the matrix $A$ is Hermitian positive semi-definite.

Let $\Phi\in\mathcal{L}(\mathcal{M}_n, \mathcal{M}_m)$. It is called positive if $\Phi(A)\geq 0$ for any $A\geq 0$. It is called $k$-positive if the linear map $\mathrm{id}_{\mathcal{M}_k}\otimes\Phi\colon \mathcal{M}_k\otimes\mathcal{M}_n \to \mathcal{M}_k\otimes\mathcal{M}_m$ is positive (where the natural isomorphism $\mathcal{M}_k\otimes\mathcal{M}_n \cong \mathcal{M}_{kn}$ is used). It is called completely positive if it is $k$-positive for any $k\in \mathbb{N}$. For a map $\Phi\colon \mathcal{M}_n \to \mathcal{M}_m$, the conditions of $n$-positivity and complete positivity are equivalent~\cite{Tomiyama_1982,Choi_Effros_1977,TAKASAKI_1982}. A linear map $\Phi\colon \mathcal{M}_n \to \mathcal{M}_m$ is called trace-preserving if $\Tr \Phi(A) = \Tr A$ for any $A\in \mathcal{M}_n$.

A linear map $\Phi\colon \mathcal{M}_n \to \mathcal{M}_m$ is called a quantum channel (or Kraus map) if $\Phi$ is completely positive and trace-preserving. We denote the set of all quantum channels from $\mathcal{M}_n$ to $\mathcal{M}_m$ as $\mathcal{C}han(n, m)$ (a notation $\mathrm{CPTP}(n, m)$ is used in some works). If $n=m$, we use the notation $\mathcal{C}han(n)$. The set of quantum channels $\mathcal{C}han(n, m)$ is a subset of the space $\mathcal{L}(\mathcal{M}_n, \mathcal{M}_m)$, and the topology on $\mathcal{C}han(n, m)$ is induced by the standard topology on the vector space $\CC^{n^2 \times m^2}$. One can verify by direct calculations that $\mathcal{C}han(n, m)$ is a convex set.

The Choi matrix $C_{\Phi}$~\cite{choi75} of a linear map $\Phi\colon \mathcal{M}_n \to \mathcal{M}_m$ is a $nm\times nm$ matrix defined by
\begin{equation*}
    C_{\Phi} = \sum_{i,j=1}^n E_{ij}\otimes \Phi(E_{ij}),
\end{equation*}
where $E_{ij}$ denotes the matrix unit ($n\times n$ matrix with $1$ in the $ij$-th entry and zeros elsewhere). In other words, the Choi matrix is a block matrix of the form
\begin{equation*}
    C_{\Phi}= 
     \begin{bmatrix} 
    \Phi(E_{11}) & \dots & \Phi(E_{1n})\\
    \vdots & \ddots & \vdots\\
    \Phi(E_{n1}) & \dots & \Phi(E_{nn})
    \end{bmatrix}.
    \end{equation*}

The mapping $\Phi \mapsto C_{\Phi}$ is a linear isomorphism between the space of linear maps from $\mathcal{M}_n$ to $\mathcal{M}_m$ and the space $\mathcal{M}_{nm}$. Choi proved~\cite{choi75} that a linear map $\Phi\colon \mathcal{M}_n \to \mathcal{M}_m$ is completely positive if and only if its Choi matrix is positive semi-definite, i.e.\ $C_\Phi\geq0$.

A complete description of Choi matrices of quantum channels can be done using the notion of the partial trace. The partial trace over the second space $\Tr_2 \colon \mathcal{M}_n\otimes\mathcal{M}_m \to \mathcal{M}_n$ is a (unique) linear map such that for all $A\in\mathcal{M}_n $ and $B\in\mathcal{M}_m$ one has $\Tr_2(A\otimes B) = \Tr(B)A$.
The map $\Phi\colon \mathcal{M}_n \to \mathcal{M}_m$ is trace-preserving if and only if $\Tr(\Phi(E_{ij}))=\Tr(E_{ij})=\delta_{ij}$, i.e. if $\Tr_2( C_{\Phi})=\mathbb{I}_n$.
Therefore, a linear map $\Phi\colon \mathcal{M}_n \to \mathcal{M}_m$ is a quantum channel if and only if $C_{\Phi}\geq 0$ and $\Tr_2(C_{\Phi})=\mathbb{I}_n$.

The set of Choi matrices for quantum channels from $\mathcal{M}_n$ to $\mathcal{M}_m$ will be denoted by $\mathcal{C}hoi(n, m)$. 
As the set $\mathcal{C}han(n,m)$ is convex and mapping $\Phi\mapsto C_{\Phi}$ is linear, the set $\mathcal{C}hoi(n,m)$ is also convex. It is easy to see that the set $\mathcal{C}hoi(n, m)$ is closed since it is the intersection of closed sets:
\begin{equation*}
\mathcal{C}hoi(n, m) = \{A \in \mathcal{M}_n\otimes\mathcal{M}_m \mid A\geq 0\} \cap\{A \in \mathcal{M}_n\otimes\mathcal{M}_m \mid \Tr_2(A)=\mathbb{I}_n\}.
\end{equation*}
The following evident statement guarantees that $\mathcal{C}hoi(n, m)$ is a bounded set. 

\begin{proposition}
    Let $C_{\Phi}=\{c_{ij}\}_{i, j=1}^{nm}$ be the Choi matrix of the quantum channel $\Phi$, $\Phi\in \mathcal{C}han(n, m)$. Then for all $i, j$ one has $|c_{ij}|\leq 1$.
\end{proposition}
\begin{proof}
One can use Sylvester's criterion for positive semi-definite matrices (see e.g.~\cite{meyer01}): the matrix is positive semi-definite if and only if all its principal minors are non-negative.
It implies that all diagonal elements of $C_{\Phi}$ are non-negative. Since $\Tr_2( C_{\Phi})=\mathbb{I}_n$, one has $0\leq c_{ii} \leq 1$. Applying  Sylvester’s criterion to the principal minor corresponding to the set of indices $\{i, j\}$, one gets $|c_{ij}|^2\leq c_{ii}c_{jj}\leq 1$.
\end{proof}

To sum up, $\mathcal{C}hoi(n, m)$ is convex, closed and bounded, therefore it is a compact convex set.

\begin{proposition}
\label{choi-chan-iso}        
The mapping $\Phi \mapsto C_{\Phi}$ is a convex-linear isomorphism of the sets $\mathcal{C}han(n, m)$ and $\mathcal{C}hoi(n, m)$. In particular, these sets are homeomorphic as topological spaces (with the topologies induced from the vector spaces), and the set $\mathcal{C}han(n, m)$ is a compact convex set.
\end{proposition}

While the case $n=m$ is often considered, important examples of quantum channels exist in $Chan(n,m)$ when $n\ne m$. Two classes of such channels are given below, the first with $n>m$ and the second with $n<m$.

1. Partial trace. The partial trace over the second space $\Tr_2\colon \M_k\otimes\M_l \to \M_k$ is a completely positive trace-preserving map and thus an example of a quantum channel from $\mathcal{C}han(kl,k)$. Similarly, the partial trace over the first space $\Tr_1\colon \M_k\otimes\M_l \to \M_l$ is an example of a quantum channel from $\mathcal{C}han(kl,l)$.

2. Quantum erasing channel. For a single qubit, the erasing channel is a channel in $Chan(2,3)$ defined by
\[
\Phi(\rho)=(1-\varepsilon)\rho+\varepsilon\Tr(\rho) |e\rangle\langle e|,
\]
where $\varepsilon\in[0,1]$ and $|e\rangle$ is the flag (erasure) state which is orthogonal to both $|0\rangle$ and $|1\rangle$. 
It leaves the qubit unchanged with probability $1-\varepsilon$ and ``erases'' it with probability $\varepsilon$. 

The phase erasing channel for a single qubit is a channel which describes that with probability $\varepsilon$ the phase of the qubit is erased without disturbing its amplitude. This quantum channel maps one-qubit to two-qubit states, i.e.\ it belongs to $\mathcal{C}han(2,4)$. In this setup, the second output qubit serves as a flag indicating whether the first qubit's phase has been randomized. This channel acts as
\[
\Phi(\rho)=(1-\varepsilon)\rho\otimes |0\rangle\langle 0|+\varepsilon\frac{\rho+\sigma_z\rho\sigma_z}{2}\otimes |1\rangle\langle 1|.
\]

\section{Kraus representation and Stiefel manifolds}
\label{Kraus_representation_and_Stiefel_manifold}

In this section, we introduce necessary preliminaries about Kraus operator-sum representation and Stiefel manifolds. We also outline our main results on the homeomorphism between the space of quantum channels and the quotient space of the Stiefel manifold by the unitary group action for any dimension $n$, as well as describe the orbits of this action.

Quantum channels of an $n$-level quantum system can be parameterized by points of some complex Stiefel manifold~\cite{PechenJPA2008, WuPechen2008}. This parameterization comes from the Kraus operator-sum representation (OSR) of a completely positive map~\cite{Kraus1983}. A linear map $\Phi\colon \mathcal{M}_n \to \mathcal{M}_m$ is completely positive if and only if for some $k\in \mathbb{N}$ there exist operators $K_l\colon \CC^n \to \CC^m$, $l\in\{1,\ldots,k\}$ such that for all $A$ in $\mathcal{M}_n$ the following holds: 
\begin{equation}
\label{cp-property}
    \Phi(A)=\sum_{l=1}^{k} K_l A K_l^{\dagger}.
\end{equation}

Operators $K_l$ are called Kraus operators, and the decomposition~\eqref{cp-property} is called the Kraus operator-sum representation. In the sum above, no more than $nm$ operators are needed. For a given completely positive map $\Phi$, the minimum number of Kraus operators required in the Kraus representation is called the Kraus rank (or the Choi rank) of $\Phi$. It coincides with the rank of the Choi matrix $C_{\Phi}$ \cite{choi75}.

The trace-preserving property of $\Phi$ is equivalent to the condition~\cite{Kraus1983} 
\begin{equation}
\label{tp-property}
    \sum_{l=1}^{k} K_l^\dagger K_l = \mathbb{I}_n.
\end{equation}
Conditions \eqref{cp-property} and \eqref{tp-property} together
are necessary and sufficient for $\Phi$ to be a quantum channel~\cite{Kraus1983}. For any quantum channel $\Phi$ we can always consider that a set of Kraus operators $(K_1,\ldots, K_k)$ from the Kraus representation~\eqref{cp-property} has exactly $nm$ operators, i.e.\ $k=nm$, possibly setting some Kraus operators to zero.
Then the Kraus representation allows quantum channels to be parameterized by points of the complex Stiefel manifold.

\begin{definition}~\cite{Stiefel1935,JamesBook}
The complex Stiefel manifold is the set 
    \begin{equation*}
    V_k(\mathbb{C}^n) = \{A\in \mathbb{C}^{n\times k} : A^{\dagger} A = \mathbb{I}_k\}
    \end{equation*}
    with the natural smooth structure induced from $\mathbb{C}^{n\times k}$.
\end{definition}

\begin{remark}
The Stiefel manifold $V_k(\mathbb{C}^n)$
can be defined as the manifold of isometries from $\mathbb{C}^k$ to $\mathbb{C}^n$ or the manifold of ordered orthonormal $k$-tuples in $\mathbb{C}^n$.
 It is a homogeneous space: $V_k(\mathbb{C}^n) \cong \mathrm{U}(n)/\mathrm{U}(n-k)$.
\end{remark}

Consider a block matrix of the size $nm^2 \times n$ associated with a particular Kraus representation
\begin{equation}
\label{K}
K = 
\begin{pmatrix}
    K_1\\
    K_2\\
    \vdots\\
    K_{nm}
\end{pmatrix}.
\end{equation}
The condition $\sum_{l=1}^{nm} K_l^{\dagger} K_l = \mathbb{I}_n$ is equivalent to $K^\dagger K = \mathbb{I}_n$. Therefore, $K$ is a point of the complex Stiefel manifold $V_n(\CC^{nm^2})$.

\begin{remark}
The Kraus representation is closely related to the Stinespring  dilation~\cite{Stinespring, Holevo_Book_2019,watrous2018theory}. In particular, a linear map $\Phi\colon \mathcal{M}_n \to \mathcal{M}_m$ is  a quantum channel if and only if for some $k$ there is an isometry $U\colon \CC^{n} \to \CC^{m}\otimes \CC^{k}$ such that
$$
\Phi(A)=\Tr_2 (UAU^{\dagger}).
$$
For a given quantum channel $\Phi$ the minimum possible number of the dimension $k$ of the environment space is its Kraus rank. If we choose $k=nm$, then the set of all quantum channels from $\mathcal{M}_n$ to $\mathcal{M}_m$ can be parameterized by isometries from $\CC^{n}$ to $\CC^{nm^2}\cong \CC^{m}\otimes \CC^{nm}$. Such isometries (so-called Stinespring isometries) are precisely defined by matrices~\eqref{K} and are points of the Stiefel manifold $V_n(\CC^{nm^2})\cong \mathrm{U}(nm^2)/\mathrm{U}(nm^2-n)$.
\end{remark}

For a quantum channel, the Kraus representation~\eqref{cp-property} with constraints~\eqref {tp-property} is non-unique~\cite{Nielsen_Chuang_2010}. Two sets of Kraus operators $(K_1, \dots, K_{nm})$ and $(\widetilde{K}_1, \dots, \widetilde{K}_{nm})$ define the same quantum channel if and only if there exists a unitary matrix $U\in\mathrm{U}(nm)$ with elements $(u_{ij})_{i,j=1}^{nm}$ such that
\begin{equation}
\label{unitary-stiefel}
    \widetilde{K}_i = \sum_{j=1}^{nm} u_{ij}K_j.
\end{equation}

Thus, the unitary group $\mathrm{U}(nm)$ acts on the complex Stiefel manifold $V_n(\CC^{nm^2})$ according to the formula \eqref{unitary-stiefel}. Orbits of this action are in one-to-one correspondence with quantum channels from $\mathcal{C}han(n, m)$. In other words, there is a bijection $h$ from the quotient space $V_n(\CC^{nm^2})/\mathrm{U}(nm)$ to $\mathcal{C}han(n, m)$.

Recall the notion of the quotient topology. Let $X$ be a topological space, $\sim$ an equivalence relation on $X$, and $\pi\colon X \to X/{\sim}$ a quotient map. The quotient topology on the set $X/{\sim}$ is the strongest topology for which the quotient map $\pi$ is continuous. Equivalently, a quotient topology on the set $X/{\sim}$ is a topology in which a subset $U\subset X/{\sim}$ is open if and only if the subset $\pi^{-1}(U)$ is open in $X$. 

A special case of the quotient space is an orbit space. If a group $G$ acts on the topological space $X$, it produces an equivalence relation "lie in the same orbit", and the topology on the space of orbits $X/G$ is defined according to the previous definition. 

We have two maps: the quotient map $\pi\colon V_n(\CC^{nm^2})\to V_n(\CC^{nm^2})/\mathrm{U}(nm)$ and the map $\widetilde{\pi}\colon V_n(\CC^{nm^2})\to \mathcal{C}han(n, m)$ defined as follows. $\widetilde{\pi}$ maps a point  $(K_1, \dots, K_{nm})$ of the Stiefel manifold to the quantum channel $\Phi$ which acts as $\Phi(\rho)=\sum_{l=1}^{nm} K_l \rho K_l^{\dagger}$. As mentioned above, there exists a natural bijection $h$ between $V_n(\CC^{nm^2})/\mathrm{U}(nm)$ and $\mathcal{C}han(n, m)$. Moreover, the following equality holds: $\widetilde{\pi}=h\circ \pi$. We establish as one of the main results of this work that $h$ is a homeomorphism.

\begin{theorem}
\label{homeomorhic}
    The topological spaces $V_n(\CC^{nm^2})/\mathrm{U}(nm)$ and $\mathcal{C}han(n, m)$ are homeomorphic, and the homeomorphism is carried out by the map $h$ described above. It corresponds to the following commutative diagram in the category of topological spaces:
    \begin{equation*}
    \def\labelstyle{\textstyle}
    \xymatrix{
    & V_n(\CC^{nm^2}) \ar[dl]_{\pi} \ar[dr]^{\widetilde{\pi}} & \\
    V_n(\CC^{nm^2})/\mathrm{U}(nm) \ar[rr]_h^{\cong}& & \mathcal{C}han(n, m)
    }
    \end{equation*}
Moreover, the maps $\pi\colon V_n(\CC^{nm^2})\to V_n(\CC^{nm^2})/\mathrm{U}(nm)$ and $\widetilde{\pi}\colon V_n(\CC^{nm^2})\to \mathcal{C}han(n, m)$ are open.   
\end{theorem}
The proof of the theorem is provided in Appendix~\ref{A}.

As $V_n(\CC^{nm^2})/\mathrm{U}(nm)$ and $\mathcal{C}han(n, m)$ are homeomorphic by the natural homeomorphism~$h$, they can be identified, so we will write $\pi\colon V_n(\CC^{nm^2})/\mathrm{U}(nm)\to \mathcal{C}han(n, m)$. 
In the next theorem, we describe the orbits of the action of the group $\mathrm{U}(nm)$ on $V_n(\CC^{nm^2})$.

\begin{theorem}
\label{Stiefel_repr}
Let $\Phi$ be a quantum channel with the Kraus rank equal to $k$. Then $\pi^{-1}(\Phi)$ is an embedded smooth submanifold of $V_n(\CC^{nm^2})$ which is diffeomorphic to the Stiefel manifold $V_{k}(\CC^{nm})$. In particular, $\pi^{-1}(\Phi)$ is diffeomorphic to $\mathrm{U}(nm)$ if and only if $k=nm$.
\end{theorem}
The proof of the theorem is provided in Appendix~\ref{B}. 

Although the bijection between the sets follows directly from the isometry between the minimal and arbitrary Kraus OSRs, the key point of this theorem is the statement about the diffeomorphism of the manifolds. The established in Theorem~\ref{homeomorhic} homeomorphism between the space of quantum channels and the quotient of the complex Stiefel manifold will be used below for the analysis of various quantum control problems.

\section{Various optimization problems over quantum channels}
\label{Optimization}
In this section, we summarize formulations of various optimization problems over quantum channels that appear in quantum control, for which in Section~\ref{Absence_of kinematic_traps} we prove the absence of local extrema using the established homeomorphism.

A typical problem in quantum control is to apply to the controlled quantum system an external action (e.g., a shaped laser field or a tailored environment) $f(t)$, which belongs to some set of admissible controls, with the goal to optimize over all admissible controls some quantity of interest, which is called objective or cost functional. If the controlled system is isolated from any environment, then it is a closed quantum system with unitary evolution operator $U_t^f$ whose dynamics is described by the Schr\"odinger equation with a control-dependent Hamiltonian $H(f)$,
\[
\frac{dU_t^f}{dt}=-iH(f)U_t^f,\quad U_0^f=\mathbb I.
\]
The dynamics of an open controlled quantum system is described by the master equation for the system density matrix $\rho_t^f$ with some control-dependent generally non-Hermitian generator ${\mathcal L}_f$, in particular in the Markovian case by the Gorini---Kossakowski---Sudarshan---Lindblad (GKSL) master equation,
\begin{equation}\label{Eq:ME}
\frac{d\rho_t^f}{dt}={\mathcal L}_f(\rho_t^f),\quad \rho_0^f=\rho_0.
\end{equation}
Non-Markovian master equations can be used as well for describing controlled open quantum systems. For an open $n$-level quantum system, the evolution of the system density matrix from $t=0$ to $t$ is described by a quantum channel, $\rho_0\to\rho_t^f=\Phi_t^f(\rho_0)$, where $\Phi_t^f\in Chan(n)$. In this way, the quantum master equation leads to an evolution in the space of quantum channels, and hence via Kraus OSR it induces a trajectory on the corresponding complex Stiefel manifold. In opposite, Kraus OSR in some cases can be used to deduce the master equation, either Markovian or not~\cite{Lidar_Bihary_Whaley_2001}.

In Mayer-type optimization problems, the objective functional depends on quantum state or state transformation (which can be unitary evolution operator for closed systems or quantum channel for open systems) of the system at some final time $T$. We consider quantum control problems with objectives which depend on the channel itself, so they are formulated as
\[
{\mathcal F}(f)=F(\Phi_T^f)\to\max \textrm{ or } \min.
\]
Here $F$ is a function on the set of quantum channels which is called kinematic objective functional, while the functional $\mathcal F$ is called dynamic objective functional. The controlled master equation~\eqref{Eq:ME} defines the map $\chi\colon f\to \Phi_T^f$ from the set of controls to the set of quantum channels. Then the objective functional $\mathcal F$ is a composition of the maps $F$ and $\chi$, $\mathcal F=F\circ \chi$. Quantum channels, which can be obtained using coherently Markovian controlled evolution, are known to form a Lie semigroup~\cite{DIRR2009, Schulte-Herbruggen2017}.

Master equation~\eqref{Eq:ME} induces channels with the same dimension of input and output spaces, i.e. channels in $Chan(n)$. However, channels with different dimensions also appear in applications. Some examples are given at the end of Sec.~\ref{Quantum_channels}, and below we provide for this more general case various classes of widely used in quantum control kinematic objective functionals $F\colon \mathcal{C}han(n,m)\to \mathbb{R}$. All the following functionals are either convex or concave.

\begin{enumerate}
    \item Expectation of a quantum system observable.

    Let $O$ be a quantum observable (a Hermitian $m\times m$ matrix)  and $\rho_0$  be the initial system density matrix ($\rho_0\in \CC^{n\times n}$, $\rho_0\geq 0$ and $\mathrm{Tr}\rho_0=1$). The objective functional which describes the problem of minimizing or maximizing the expectation of the observable $O$ is determined by the kinematic functional
    \begin{equation*}
        F_O(\Phi)=\Tr\left(O \Phi(\rho_0)\right),\quad \Phi\in\mathcal{C}han(n,m).
    \end{equation*}
     This kinematic objective functional is an affine function on $\mathcal{C}han(n,m)$, hence it is both convex and concave.

\item Thermodynamic and informational characteristics. 

Let $S$ be some quantum mechanical entropy (von Neumann, quantum conditional entropy, quantum Tsallis, Renyi, Renner’s (conditional) min- and max-entropies, Petz's quasi-entropies, etc.)~\cite{Ohya_Petz_Book_1993, Ohya_Watanabe_Book_2010}. We consider those which are concave functions on the set of density matrices. For example, von Neumann entropy has the form
\[
S(\rho)=-\Tr(\rho\log\rho)=-\sum_{\lambda_i\neq 0}\lambda_i\log{\lambda_i},
\]
where $\{\lambda_i\}$ are eigenvalues of the density matrix $\rho$.
The control problem of maximizing or minimizing quantities involving entropy is determined by the kinematic objective functional~\cite{PechenRabitz2010}
\begin{equation*}
F_{O}^{S}(\Phi)=-\Tr (O\Phi(\rho_0)) +\frac 1\beta S(\Phi(\rho_0)), \quad \Phi\in\mathcal{C}han(n,m),
\end{equation*}
where $\beta>0$. It determines an important physical class of objectives that include thermodynamic quantities such as free energy, which corresponds to $O=H$, where $H$ is the Hamiltonian of the system. Other entropy involving quantities can also be considered~\cite{MorzhinPechen2024}. The entropy function of a quantum channel can be defined by the entropy of its Choi matrix~\cite{Chu_Huang_Li_Zheng_2022}. From the continuity of the von Neumann entropy on the set of all density matrices (see e.g.~\cite{Nielsen_Chuang_2010, Shirokov2010}), the continuity of the functional $F_{O}^{S}$ follows. This kinematic objective functional is a concave function on $\mathcal{C}han(n,m)$.

\item Quantum channel generation functional.

A problem  of generation of a quantum channel $\Phi_0\in \mathcal{C}han(n,m)$ can be formulated as a minimization problem with the objective functional determined by the kinematic objective functional
    \begin{equation*}
        F_{\Phi_0}(\Phi)=\|\Phi-\Phi_0\|_F^2,\quad\text{where } \|\mathord{\cdot}\|_F \text{ is Frobenius norm,}\quad\Phi\in\mathcal{C}han(n,m).
    \end{equation*}
We can also consider other norms such as a diamond norm~\cite{Aharonov_Kitaev_Nisan_1998,watrous2018theory} or the distance between quantum channels which is considered in Sec.~\ref{sec:metric} of this work. This functional (with any norm in the definition) is convex and attains its minimum at the single point $\Phi=\Phi_0$. Examples of quantum channels with $n\ne m$ are given above.

A special case of this functional includes quantum gate generation functionals. A quantum gate is the unitary operator $W\in U(n)$ defined up to a physically irrelevant phase.  Suppose $\Phi_W\in\mathcal{C}han(n)$ is a unitary quantum channel corresponding to the quantum gate $W$, so it acts on any density matrix $\rho$ as $\Phi_W(\rho)=W\rho W^{\dagger}$. Then the quantum gate generation functional  has the form
\begin{equation*}
    F_W(\Phi)=\|\Phi-\Phi_W\|_F^2,\quad \Phi\in\mathcal{C}han(n).
\end{equation*}

\item Quantum gate generation functional with three density matrices.
    
Let $\rho_1=\sum_{k=1}^n \lambda_k|\phi_k\rangle \langle \phi_k|$, where $\{|\phi_k\rangle\}_{k=1}^n$ is an orthogonal basis in $\mathbb{C}^n$,  be a~density matrix with $n$ different nonzero eigenvalues. Let $\rho_2$ be a~one dimensional orthogonal projector on $\mathbb{C}^n$ such that $\rho_2|\phi_k\rangle\neq 0$ for $k= 1,\ldots,n$. Let $\rho_3=\frac{1}{n}\mathbb{I}_n$. Then another functional for the problem of generating a quantum gate is:
\begin{equation*}
F^{\mathrm{GRK}}_W(\Phi) := \frac{1}{6}\sum_{i=1}^3 \|\Phi(\rho_i)-W\rho_i W^{\dagger}\|_F^2,\; \Phi\in\mathcal{C}han(n).
\end{equation*}
The functional $F_W^{\mathrm{GRK}}$ is equal to 0 if and only if $\Phi(\rho_i)=W\rho_iW^\dagger$ for $i = 1,2,3$.  As proved by Goerz, Reich, and Koch in~\cite{goerz2014, Goerz_2021}, if these three equalities are satisfied for a quantum channel $\Phi$  then $\Phi(\rho)=W \rho W^\dagger$ for any density matrix $\rho$, i.e.\ $\Phi=\Phi_U$.  Thus, the minimum value of this functional is achieved at the required target unitary quantum channel. This functional is convex.
\end{enumerate}

Examples of finding the infimum over quantum channels from $\mathcal{C}han(n,m)$ acting between spaces of different dimensions ($n\ne m$) include squashed information and conditional entanglement of mutual information~\cite{Hayden_Jozsa_Petz_Winter_2004, Yang_Horodecki_Wang_2008}.

All quantum control problems above are described using objective functionals defined on the topological space of quantum channels. Alternatively, they can be described by objective functionals on the corresponding complex Stiefel manifold, since any kinematic objective functional $F\colon\mathcal{C}han(n,m)\to\RR$
generates the (kinematic) objective functional $\widetilde{F}:=F\circ \pi$ on the complex Stiefel manifold  $V_n(\CC^{nm^2})$. One can consider the problem of maximizing or minimizing objective functionals on the Stiefel manifolds and use Riemannian optimization methods. The corresponding theory was developed in~\cite{PechenJPA2008, WuPechen2008, OzaJPA2009, Luchnikovetc2021a}. Randomized gradient descents on Riemannian manifolds for quantum optimization in another context were studied in~\cite{MalvettiArenz2024}. Note that Stiefel manifolds are homogeneous spaces. Optimization problems on homogeneous spaces were considered in the literature~\cite{Sachkov, Jurdjevic}. In this regard, it would be interesting to study the lifting of the controlled Markovian dynamics to the Stiefel manifolds. In the next section, we prove the absence of local but not global extrema for the objective functionals $\widetilde{F}$ over Stiefel manifolds which are induced by the kinematic objective functionals for the four quantum control problems listed above and even more, for the objective functional $\widetilde{F}$ induced by any convex kinematic objective functional on the set of quantum channels.

\section{Control landscapes and absence of traps}
\label{Absence_of kinematic_traps}

In this section, we apply the established in Sec.~\ref{Kraus_representation_and_Stiefel_manifold} homeomorphism to the analysis of kinematic objective functionals defined on the space of quantum channels and on the complex Stiefel manifolds including objective functionals describing outlined in Sec.~\ref{Optimization} control problems for open quantum systems. For such control problems, we prove the absence of kinematic traps (local optima that are not global) for corresponding objective functionals defined on the complex Stiefel manifolds.

For this, we first show that the quotient map transforms local extremum points of an objective functional on the Stiefel manifold to local extremum points of the corresponding objective functional on the set of quantum channels. Then we apply this result to solve a non-trivial problem of proving the absence of kinematic traps for objective functionals defined on the complex Stiefel manifolds. The absence of traps is proved here for a broad class of quantum objective functionals defined on the Stiefel manifolds, including the mean value of a quantum observable, various functionals describing the generation of a quantum channel, and, more generally, any convex objective functionals for an $n$-level open quantum systems of an arbitrary dimension~$n$.

The importance of the analysis of quantum control landscapes is motivated by the common use in the practice of local search algorithms for finding optimal solutions for quantum control problems. The efficiency of local search algorithms strongly depends on the presence or absence of traps~\cite{Rabitz2004, Fouquieres2013} that makes important the analysis of the presence or absence of traps for quantum control problems. To fully carry out this analysis, it is necessary to analyze the dynamic objective functional ${\mathcal F}(f)$. Since ${\mathcal F}(f)$ is determined by the kinematic functional $F(\Phi)$, the properties of the kinematic control landscapes on the set of quantum channels $\mathcal{C}han(n,m)$ or on the complex Stiefel manifold $V_{n}(\CC^{nm^2})$ are essential for the dynamic objective functional. 
\begin{definition}
The kinematic control landscape is the graph of the kinematic objective functional $F\colon\mathcal{C}han(n,m)\to\RR$ or the associated functional $\widetilde{F}\colon V_{n}(\CC^{nm^2})\to\RR$. A kinematic trap for the problem of maximizing (minimizing) the objective functional is a point of local but not global maximum (minimum) of the kinematic objective functional $F$ or $\widetilde{F}$.
\end{definition}

So, a trap is a suboptimal maximum (for maximization) or minimum (for minimization) of the objective functional. In the same way, we can define the control landscape on the set of Choi matrices. However,  each continuous functional $F$ on the set $\mathcal{C}han(n, m)$ can be associated with the continuous functional $F_{Choi}$ on the set $\mathcal{C}hoi(n, m)$ via the relation $F_{Choi}(C_{\Phi}) = F(\Phi)$. Then by Proposition~\ref{choi-chan-iso}, extrema points of $F_{Choi}$ and $F$ are in one-to-one correspondence: if $\Phi_0$ is the extremum point of $F$, then $C_{\Phi_0}$ is the extremum point of $F_{Choi}$ of the same type. Thus the landscapes on the set of matrices Choi and on the set of quantum channels are equivalent in the sense that they have the same extrema points.
From the compactness of $\mathcal{C}hoi(n, m)$ then it follows that any continuous kinematic functional on $\mathcal{C}hoi(n, m)$, and hence on $\mathcal{C}han(n, m)$, reaches its global extremum.

As a corollary of Theorem~\ref{homeomorhic}, we obtain the following connection between extrema points of control landscapes on the topological space of quantum channels $\mathcal{C}han(n, m)$ and on the complex Stiefel manifold $V_{n}(\CC^{nm^2})$.

\begin{theorem}
\label{extrema-match}
Let $F\colon \mathcal{C}han(n, m)\to \RR$ be a continuous function on the space of quantum channels and  $\widetilde{F}=F\circ\pi$ be the corresponding function on the complex Stiefel manifold  $V_n(\CC^{nm^2})$. Then $s$ is a point of a non-strict local extremum of the function $\widetilde{F}$ on the $V_n(\CC^{nm^2})$ if and only if $\pi(s)$ is a point of non-strict local extremum of the function $F$ on the space of quantum channels $\mathcal{C}han(n, m)$.
\end{theorem}

\begin{proof}
In this proof, by extremum we mean a non-strict extremum. Let $s_0$ be a local extremum point of a functional $\widetilde{F}$ on the complex Stiefel manifold, i.e.\ $s_0$ is an extremum point in some open neighbourhood $\widetilde{U}\subset V_n(\CC^{nm^2})$. Since $\pi$ is an open map, $\pi(\widetilde{U})$ is an open neighborhood of $\pi(s_0)$ in $\mathcal{C}han(n, m)$. As $\widetilde{F}(s)=F(\pi(s))$ for all $s$, then $\pi(s_0)$ is an extremum point in $\pi(\widetilde{U})$. Conversely, let $\pi(s_0)$ be an extremum point in some open neighbourhood $U\subset \mathcal{C}han(n, m)$. Then $s_0$ is an extremum point in $\pi^{-1}(U)$, which is open as $\pi$ is continuous.
\end{proof}

For a controllable $n$-level closed quantum system, the kinematic objective functional is defined on the corresponding Lie group, the special unitary group $\mathrm{SU}(n)$. It is known that control landscapes for such functionals have no kinematic traps~\cite{Neumann1937, Brockett1989, Rabitz2004}.  

Theorem~\ref{extrema-match} allows us to reduce the study of local extrema points from the complex Stiefel manifold to the space of quantum channels. For convenience, relevant information from convex analysis is provided in Appendix \ref{C}. Using the fact that a convex function on a convex set does not have local but not global minima, we obtain the following theorem.

\begin{theorem}
    Let $F\colon \mathcal{C}han(n,m)\to \RR$ be a continuous convex (concave) functional and $\widetilde{F}=F\circ\pi$ be a corresponding functional on the complex Stiefel manifold $V_n(\CC^{nm^2})$. Then $F$ has no local but not global minima (maxima) points on $\mathcal{C}han(n,m)$ and $\widetilde{F}$ has no local but not global minima (maxima) points on $V_n(\CC^{nm^2})$.
\end{theorem}

We may apply this theorem to the functionals of interest described in Sec.~\ref{Optimization}. The functional $F_O$ is both convex and concave, so it has no local but not global maxima or minima on $\mathcal{C}han(n,m)$. Thus $\widetilde{F}_O$ has no kinematic traps on the Stiefel manifold $V_n(\CC^{nm^2})$. 

Similarly, the functional $F_O^S$ is concave and the functionals $F_{\Phi_0}$, $F_{W}$ and $F_W^{\mathrm{GRK}}$ are convex. Thus, $\widetilde{F}_O^S$ has no points of local but not global maxima and $\widetilde{F}_{\Phi_0}$, $\widetilde{F}_{W}$ and $\widetilde{F}_W^{\mathrm{GRK}}$ have no points of local but not global minima on the corresponding Stiefel manifold. As for local minima for $F_O^S$ and local maxima for $F_{\Phi_0}$, $F_{W}$ and $F_W^{\mathrm{GRK}}$, they are attained at the extreme points of $\mathcal{C}han(n, m)$. It is, in itself, interesting and non-trivial to investigate the extreme points of the space of quantum channels~\cite{FRIEDLAND2016553}, and the existence of local but non-global extrema in this case requires further analysis.

We now look at the set of points on the Stiefel manifold at which the extrema of the functionals are attained. The global minimum of $F_{\Phi_0}$, $F_{W}$ and $F_W^{\mathrm{GRK}}$ is attained at a quantum channel $\Phi_0$ or $\Phi_W$ respectively. Therefore, the subset of the Stiefel manifold at which the minimum is attained is an orbit $\pi^{-1}(\Phi_0)$ or $\pi^{-1}(\Phi_W)$ respectively. According to Theorem~\ref{Stiefel_repr}, this orbit is an embedded smooth submanifold which is diffeomorphic to $V_k(\CC^{nm})$ for $F_{\Phi_0}$ or $V_k(\CC^{n^2})$ for $F_{W}$, where $k$ is the Kraus rank of this channel. As the quantum channel $\Phi_W$ has the Kraus rank equal to $1$,  the minimum subset for $\widetilde{F}_W$ and $\widetilde{F}_W^{\mathrm{GRK}}$ is an embedded smooth submanifold diffeomorphic to $V_1(\CC^{n^2})\cong S^{2n^2-1}$.

The sets of points at which the extrema of functionals $F_O$ and $F_O^S$ are attained, as well as maxima sets for functionals $F_{\Phi_0}$, $F_{W}$ and $F_W^{\mathrm{GRK}}$ are more complicated to analyze and require further consideration.

In the kinematic quantum control landscape analysis, it is commonly assumed that any kinematic map (i.e.\ any unitary evolution operator for a closed quantum system and any quantum channel for an open quantum system) can be generated using available physical controls via Schr\"odinger or master equation, either Markovian or non-Markovian. It is known that quantum channels which can be obtained using coherently Markovian controlled evolution form a Lie semigroup~\cite{DIRR2009, Schulte-Herbruggen2017}.  In common situations, not all density matrices can be generated using the physical controls in the dynamic master equations (e.g.,~\cite{Altafini2003, Lokutsievskiy_Pechen_2021}), and hence not all quantum channels can be generated. In such situations, one has first to describe the set of reachable quantum channels, and then analyze the corresponding constrained quantum control landscape which is defined not over the whole set of quantum channels but over some subset which can be physically generated. For this purpose, methods of geometric control theory are usually used.

\section{Stiefel-induced metric on the space of quantum channels}
\label{sec:metric}

The analysis in the previous section was performed using the established homeomorphism between the topological space of quantum channels and the quotient of the complex Stiefel manifold under the action of the unitary group. In this section, we show that this action specifies a natural metric on the topological space of quantum channels. We show that this metric coincides with the Bures angle in the case of density matrices.

The canonical Riemannian metric on the Stiefel manifold $V_n(\CC^{nm^2})$, i.e. the metric of the homogeneous space $\mathrm{U}(nm^2)/\mathrm{U}(nm^2-n)$, defines the structure of a metric space and the distance $d_{\mathrm{can}}(\cdot,\cdot)$ between points of this manifold. This metric induces a metric on $V_n(\CC^{nm^2})/\mathrm{U}(nm)$ and, hence, a metric $D_{\mathrm{can}}(\cdot,\cdot)$ on the space of quantum channels $\mathcal{C}han(n, m)$. The corresponding distance between any two channels $\Phi_1,\Phi_2\in\mathcal{C}han(n, m)$ is defined by the formula
\begin{equation*}
D_{\mathrm{can}}(\Phi_1,\Phi_2)=\inf\{d_{\mathrm{can}}(s_1,s_2)\colon s_1\in \pi^{-1}(\Phi_1),s_2\in \pi^{-1}(\Phi_1)\}.
\end{equation*}
Along with the canonical metric on the Stiefel manifold $V_n(\CC^{nm^2})$, one can consider the Euclidean metric $d_{\mathrm{Eucl}}(\cdot,\cdot)$ induced by the embedding in the Euclidean space (as was considered for real Stiefel manifolds in~\cite{Edelman1998}). This metric also induces a metric on the space of quantum channels:
\begin{equation*}
D_{\mathrm{Eucl}}(\Phi_1,\Phi_2)=\inf\{d_{\mathrm{Eucl}}(s_1,s_2)\colon s_1\in \pi^{-1}(\Phi_1),s_2\in \pi^{-1}(\Phi_1)\}.
\end{equation*}
Other Riemannian metrics on the Stiefel manifold and corresponding metrics on the space of quantum channels can also be considered.

Below we consider the case $n=1$. In this case, $\mathcal{C}han(1, m)$ coincides with the space of $m\times m$ density matrices. We will show that in this case the distance between the orbits coincides with the Bures angle \cite{bengtsson2007geometry, Nielsen_Chuang_2010} which is defined using the fidelity of two quantum states \cite{uhlmann1976transition, uhlmann2011fidelity}.
The fidelity $F$ and the Bures angle $D_A$ between two quantum states $\rho_1$ and $\rho_2$ are defined as
\begin{eqnarray*}
    F(\rho_1, \rho_2) &:=& \left(\Tr\sqrt{\sqrt{\rho_1}\rho_2\sqrt{\rho_1}}\right)^2\, ,\\
    D_A(\rho_1, \rho_2) &:=& \arccos\sqrt{F(\rho_1, \rho_2)} = \arccos \Tr\sqrt{\sqrt{\rho_1}\rho_2\sqrt{\rho_1}}.
\end{eqnarray*}

The Stiefel manifold $V_1(\CC^{N})$ is the $(2N-1)$-dimensional sphere $S^{2N-1}$. The group $\mathrm{U}(m)$ acts on $V_1(\CC^{m^2})=S^{2m^2-1}$ in the following way. Let
$v = (v_1, \dots, v_m) \in S^{2m^2-1}$ be a point on the unit sphere in $\CC^{m^2}$, $v_i \in \CC^m$, $\sum_{i=1}^m \|v_i\|^2=1$. We can associate the vector $v\in S^{2m^2-1}$ with the vector $\sum_{i=1}^m |v_i\rangle\otimes |e_i\rangle \in \CC^m\otimes\CC^m$, where $\{|e_i\rangle\}_{i=1}^m$ is the standard orthonormal basis in $\CC^m$ and for the vector $v_i\in\CC^m$ we also use the Dirac notation $|v_i\rangle$. Formula \eqref{unitary-stiefel} implies that the action of an element $U\in\mathrm{U}(m)$ corresponds to the action of $\I\otimes U$ in $\CC^m\otimes \CC^m$. Formula \eqref{cp-property} implies that a point $v\in S^{2N-1}$ corresponds to the density matrix $\rho$ if and only if $\rho = \sum_{i=1}^m |v_i\rangle \langle v_i|$. In fact, the vector~$v$ corresponds to the purification of the state $\rho$ \cite{Holevo_Book_2019, bengtsson2007geometry}.
We now prove that the induced metric $D_{\mathrm{Eucl}}$ on $\mathcal{C}han(1, m)$ coincides with the Bures angle.
\begin{theorem}
    Let $\rho_1$ and $\rho_2$ be any two $m\times m$ density matrices. Then $D_{\mathrm{Eucl}}(\rho_1,\rho_2)= D_A(\rho_1, \rho_2)$.
\end{theorem}
\begin{proof}
Let $\rho_1 = \sum_{i=1}^m \lambda_i |\psi_i\rangle\langle\psi_i|$ and $\rho_2 = \sum_{i=1}^m \mu_i |\varphi_i\rangle\langle\varphi_i|$, where $\{|\psi_i\rangle\}_{i=1}^m$ and $\{|\varphi_i\rangle\}_{i=1}^m$ are orthonormal eigenbases. Then one has
\begin{equation*}
    B := \sqrt{\rho_1}\rho_2\sqrt{\rho_1} = \sum_{i, j, k=1}^m \sqrt{\lambda_i\lambda_k} \mu_j \langle\psi_i|\varphi_j\rangle\langle\varphi_j|\psi_k\rangle |\psi_i\rangle\langle\psi_k|.
\end{equation*}
Thus, the operator $B$ has the matrix elements 
\begin{equation*}
    B_{ik} = \sum_{j=1}^m \sqrt{\lambda_i\lambda_k}\mu_j \langle\psi_i|\varphi_j\rangle\langle\varphi_j|\psi_k\rangle
\end{equation*}
in the basis $\{|\psi_i\rangle\}$.
Let $\orb(\rho_1), \orb(\rho_2)$ be the $\mathrm{U}(m)$-orbits on the sphere $S^{2m^2-1}$ corresponding to $\rho_1$ and $\rho_2$ respectively. Since 
$$\rho_1 = \sum_{i=1}^m \lambda_i |\psi_i\rangle\langle\psi_i| = \sum_{i=1}^m (\sqrt{\lambda_i} |\psi_i\rangle)(\sqrt{\lambda_i}\langle\psi_i|) = \sum_{i=1}^m |v_i\rangle \langle v_i|,
$$
where $|v_i\rangle = \sqrt{\lambda_i} |\psi_i\rangle$, the vector $v = \sum_{i=1}^m |v_i\rangle \otimes |e_i\rangle = \sum_{i=1}^m\sqrt{\lambda_i}|\psi_i\rangle\otimes|e_i\rangle$ lies in $\orb(\rho_1)$. Analogously $v' = \sum_{i=1}^m\sqrt{\mu_i}|\varphi_i\rangle\otimes|e_i\rangle$ belongs to $\orb(\rho_2)$. The distance between any vectors $w$ and $w'$ on the unit sphere is $d_{\mathrm{Eucl}}(w, w') = \arccos \re \langle w, w'\rangle$, where $\langle\mathord{\cdot}, \mathord{\cdot}\rangle$ is the standard scalar product in $\CC^{m^2}$. Then the distance between the orbits $\orb(\rho_1)$ and $\orb(\rho_2)$ is 
\begin{multline*}
    D_{\mathrm{Eucl}}(\rho_1, \rho_2) = \min\limits_{U_1, U_2 \in\mathrm{U}(m)} \arccos\re\langle (\I\otimes U_1)v, (\I\otimes U_2)v'\rangle \\= \arccos\max\limits_{U\in\mathrm{U}(m)}\re\langle v, (\I\otimes U)v'\rangle.
\end{multline*}
Rewriting the expression on the right-hand side of the above equality using matrix elements of the objects appearing in this expression gives
\begin{equation*}
    D_{\mathrm{Eucl}}(\rho_1, \rho_2) = \arccos \max_{U\in\mathrm{U}(m)} \re \sum_{i, j=1}^m \sqrt{\lambda_i\mu_j}U_{ij}\langle\psi_i|\varphi_j\rangle,
\end{equation*}
where $U_{ij}$ are the matrix elements of operator $U$ in the basis $\{|\psi_i\rangle\}$. Let $A$ be the operator with the matrix elements $A_{ji} = \sqrt{\lambda_i\mu_j}\langle\psi_i|\varphi_j\rangle$ in this basis. Then one has
$$
\sum_{i, j=1}^m \sqrt{\lambda_i\mu_j}U_{ij}\langle\psi_i|\varphi_j\rangle = \sum_{i, j=1}^m U_{ij}A_{ji} = \Tr(U A).
$$ 
Moreover, $A^\dagger A = B$, since
$$
(A^\dagger A)_{ik} = \sum_{j=1}^m A_{ji}^* A_{jk} = \sum_{j=1}^m\sqrt{\lambda_i\lambda_k}\mu_j \langle\psi_i|\varphi_j\rangle \langle\varphi_j|\psi_k\rangle = B_{ik}\, .
$$ 
Thus the polar decomposition of the operator $A$ \cite{horn2013} has the form $A = U' \sqrt{B}$ for some unitary operator $U'$. Here $\sqrt{B}$ is the non-negative square root of the operator~$B$. Hence, 
$$D_{\mathrm{Eucl}}(\rho_1, \rho_2) = \arccos \max\limits_{U\in\mathrm{U}(m)}\re \Tr(U\sqrt{B}).$$
Since $D_A(\rho_1, \rho_2) = \arccos\Tr \sqrt{B}$, for proving the theorem it is sufficient to prove that $\max\limits_{U\in\mathrm{U}(m)}\re \Tr(U\sqrt{B}) = \Tr \sqrt{B}$. Since the operator $\sqrt{B}$ is Hermitian and positive, it has the form $\sqrt{B} = V\Lambda V^\dagger$, where $\Lambda = \diag(p_1, \dots, p_m)$, $p_i\geqslant 0$, $V$ is a unitary operator. Then 
\begin{equation*}
    \re \Tr(U\sqrt{B}) = \re \Tr(U V\Lambda V^\dagger) = \re \Tr(V^\dagger U V\Lambda) = \re \Tr(V' \Lambda),
\end{equation*}
where $V' = V^\dagger U V$. Let $V'$ have the matrix elements $V'_{ij}$, then
\begin{equation*}
    \re \Tr(V' \Lambda) = \re \sum_{i=1}^m V'_{ii} p_i \leq \left|\sum_{i=1}^m V'_{ii} p_i\right| \leq  \sum_{i=1}^m \left|V'_{ii} p_i\right| \leq \sum_{i=1}^m p_i = \Tr(\sqrt{B}),
\end{equation*}
where the third inequality holds due to $|V'_{ij}|\leq 1$. The maximal value $\Tr(\sqrt{B})$ is attained at $V'=\mathbb{I}$, that is $U=VV^\dagger = \mathbb{I}$. Thus, $D_{\mathrm{Eucl}}(\rho_1, \rho_2) = D_A(\rho_1, \rho_2)$.
\end{proof}

Thus, the Stiefel-induced metric on the space of quantum channels is a generalization of the Bures angle for density matrices. A closely related concept to the Bures angle is the Bures distance, which is defined for density matrices and quantum channels~\cite{Holevo_Book_2019, Wilde_Book_2017}. It determines the minimum possible distance induced by the operator norm between the Stinespring isometries of channels, and it can also be defined as some distance between orbits.  The continuity of Stinespring isometries from a quantum channel was proved in finite-dimensional and infinite-dimensional cases~\cite{KRETSCHMANNa, KRETSCHMANN}. In these works, it was proved that for quantum channels which are sufficiently close with respect to the diamond norm metric~\cite{Aharonov_Kitaev_Nisan_1998,watrous2018theory}, there are corresponding Stinespring isometries that have the same environment and are sufficiently close with respect to the operator norm. Such continuity was further studied in \cite{Shirokov2020, Shirokov2022, vomEnde2023, vomEnde2023a} both in finite- and infinite-dimensional cases. Some analogue of the Bures angle for quantum channels was studied in~\cite{Yuan2017}. The geometry of the complex Stiefel manifold was not considered in the mentioned works, which might be an interesting question for further research. It would also be interesting to compare the Stiefel-induced metric with other known distances on the space of quantum channels, such as the diamond norm metric and the Bures distance.

\section{Conclusions}
\label{Conclusions}

The paper studies the relationship between two different parametrizations of quantum channels: using Choi matrices and using a complex Stiefel manifold. It is proved that the bijection between the quotient of the complex Stiefel manifold with respect to the action of the unitary group and the space of quantum channels is a homeomorphism. The orbit structure of the unitary group's action on the Stiefel manifold is described. As an application of the obtained homeomorphism
to the study of kinematic landscapes for open quantum systems, we prove that the surjection of the Stiefel manifold onto the space of quantum channels maps local extremum points of a continuous function to local extremum points. Then using the methods of convex analysis, the absence of suboptimal traps on the Stiefel manifold for an $n$-level open quantum system is proved for the functionals describing the mean value of the quantum observable, thermodynamic quantities involving entropy, the functional for generating a quantum channel, and the functional for generating a quantum gate on three states. Finally, using the parametrization of quantum channels by points of the quotient of the complex Stiefel manifold and the Riemannian metric on the Stiefel manifold, we define the corresponding metric on the space of quantum channels. In the case of density matrices, we show that this metric coincides with the Bures angle. 

\medskip

\noindent {\bf Acknowledgements}
This work was supported by the Ministry of Science and Higher Education of the Russian Federation (Grant No. 075-15-2024-529).

\medskip

\noindent {\bf Data availability} No datasets were generated or analysed during the current study.

\section*{Declarations}
{\bf Conflict of interest} The authors declare that they have no conflict of interest.

\appendix
\section{Proof of Theorem~\ref{homeomorhic}}
\label{A}

First, we formulate some auxiliary facts. Choi's work~\cite{choi75} showed the following relationship between a completely positive linear map and its Choi matrix. Let $\Phi\colon \mathcal{M}_n \to \mathcal{M}_m$ be a completely positive map and $C_{\Phi}$  be its Choi's matrix, which is written in the form
    \begin{equation}
    \label{choi-matrix-decomposition}
        C_{\Phi} = \sum_{l=1}^{nm} v_l v_l^{\dagger},
    \end{equation}
    where $\{v_l\}_{l=1}^{nm}$ is a set of vectors from $\CC^{nm}$ (some of these vectors may be zero).
    For the vectors $v_l$, consider $m\times n$ matrices $K_l$ defined as follows:
    \begin{equation}
    \label{vector-matrix-matching}
    v_l = 
    \begin{pmatrix}
        v_{l,1} \\
        v_{l,2} \\
        \vdots \\
        v_{l,nm}
    \end{pmatrix}
    \quad \mapsto \quad
    K_l = 
    \begin{pmatrix}
        v_{l,1} & v_{l,m+1} & \dots & v_{l, nm-m+1} \\
        v_{l,2} & v_{l,m+2} & \dots & v_{l, nm-m+2}\\
        \vdots & \vdots & \ddots & \vdots \\
        v_{l,m} & v_{l,2m} & \dots & v_{l, nm}
    \end{pmatrix}
    \end{equation}
Then the map $\Phi$ has the form
    \begin{equation*}
        \Phi(A)=\sum_{l=1}^{nm} K_l A K_l^{\dagger},
    \end{equation*}
and moreover,~\eqref{vector-matrix-matching} performs a bijection between Choi matrix decompositions of the form~\eqref{choi-matrix-decomposition} and Kraus representations of $\Phi$ with exactly $nm$ Kraus operators~\cite{choi75}.

Let us now consider the $nm\times nm$ matrix $$X=
\begin{pmatrix}
    v_{1, 1} & v_{2, 1} & \dots & v_{nm, 1}\\
    v_{1, 2} & v_{2, 2} & \dots & v_{nm, 2}\\
    \vdots & \vdots & \ddots & \vdots\\
    v_{1, nm} & v_{2, nm} & \dots & v_{nm, nm}
\end{pmatrix}.
$$
We have
\begin{equation*}
C_{\Phi}=\sum_{i=1}^{nm}v_iv_i^{\dagger}=\sum_{i,j,k=1}^{nm}v_{i,j}\overline{v}_{i,k}E_{jk},
\end{equation*}
or in other words $(C_{\Phi})_{jk}=\sum_{i=1}^{nm}v_{i,j}\overline{v}_{i,k}$.  This implies $X X^{\dagger}=C_{\Phi}$. Thus, one has the following proposition. 

\begin{proposition}
\label{kraus-as-root}
Let $\Phi$ be a quantum channel. Then there is a one-to-one correspondence between the Kraus expansions defining this quantum channel
\begin{equation*}
\Phi(\rho)=\sum_{l=1}^{nm}K_l\rho K_l^{\dagger},
\end{equation*}
and solutions of the equation $X X^{\dagger}=C_{\Phi}$ with respect to $X\in\CC^{nm\times nm}$. 
\end{proposition}

If $A\geq 0$ is a positive semi-definite matrix, by $\sqrt{A}$ denote the non-negative square root of  $A$. Then $\sqrt{A}$ is a unique matrix $B$ that is positive semi-definite and such that $B^2=A$. Let $\|\mathord{\cdot}\|$ denote the operator norm. Note that if $A,B$ are positive semi-definite operators, then (see e.g.~\cite{Bhatia})
\begin{equation}
\label{sqrt_inequality}
\|\sqrt{A}-\sqrt{B}\|\leq \sqrt{\|A-B\|}.
\end{equation}

Now we proceed to the actual proof of Theorem~\ref{homeomorhic}.
\begin{proof}
All the spaces under consideration satisfy the first axiom of countability. Therefore, to prove the continuity of maps, it is sufficient to prove sequential continuity.
Let us show that the map $\widetilde{\pi}$ is continuous. Consider a sequence of points $\{s_k\}_{k=1}^{\infty}$ of the Stiefel manifold $V_n(\CC^{nm^2})$ which converges to the point $s\in V_n(\CC^{nm^2})$. Let $s_k=(K_1^k, \dots, K_{nm}^k)$, where $k$ is the superscript, and $s=(K_1, \dots, K_{nm})$. Let $\Phi_k$ be a sequence of quantum channels corresponding to $(K_1^k, \dots, K_{nm}^k)$ and $\Phi$ be a quantum channel corresponding to $(K_1, \dots, K_ {nm})$. Then, as unit matrices $E_{ij}$ form an orthonormal basis with respect to the Hilbert-Schmidt  inner product, we get
\begin{equation}
\label{diff-norm}
\|\Phi_k - \Phi\|_F^2 = \sum_{i,j=1}^{n} \|\Phi_k(E_{ij}) - \Phi(E_{ij})\|_F^2 = \sum_{i,j=1}^{n} \left\lVert \sum_{l=1}^{nm} (K_l^k E_{ij} K_l^{k\dagger} - K_l E_{ij} K_l^{\dagger})\right\rVert_F^2.
\end{equation}
Since $s_k\to s$ as $k\to\infty$, we have $K_l^k \to K_l$ as $k\to \infty$ for all $l\in\{1,\ldots,mn\}$. Hence,  
\begin{equation*}
K_l^k E_{ij} K_l^{k\dagger} - K_l E_{ij} K_l^{\dagger} = (K_l^k - K_l)E_{ij} K_l^{k\dagger} + K_l E_{ij}(K_l^{k\dagger} - K_l^{\dagger}) \to 0 \quad \text{as } k \to \infty.
\end{equation*}
Therefore, in \eqref{diff-norm} each term of the internal sum tends to 0. Hence $\|\Phi_k-\Phi\|_F^2 \to 0$ for $k\to \infty$. This means that the map $\widetilde{\pi}\colon V_n(\CC^{nm^2})\to \mathcal{C}han(n, m)$ is continuous.

So we have that $\widetilde{\pi} = h \circ \pi$ is continuous. We will use one more fact about the quotient topology, namely, the universal property of the quotient topology~\cite{Bradley2020}.
Let $X$ be a topological space and $\pi\colon X \to X/{\sim}$ be the quotient map onto the corresponding quotient space. Let $f\colon X/{\sim} \to Z$ be some map, where $Z$ is an arbitrary topological space. Then $f$ is continuous if and only if $f\circ \pi\colon X \to Z$ is continuous. So, by the universal property of the quotient topology, the map $h\colon V_n(\CC^{nm^2})/\mathrm{U}(nm) \to \mathcal{C}han (n, m)$ is continuous. Let us prove that it is a homeomorphism. To do this, it suffices to show that $h^{-1}$ is also continuous.
    
Consider a sequence of quantum channels $\{\Phi_k\}$ which converges to a quantum channel $\Phi$. Then the sequence of Choi matrices $\{C_{\Phi_k}\}$ converges to $C_{\Phi}$. Inequality~\eqref{sqrt_inequality} implies that $\sqrt{C_{\Phi_k}} \to \sqrt{C_{\Phi}}$ as $k\to \infty$. Let $s_k=(K_1^k, \dots, K_{nm}^k)$ and $s=(K_1, \dots, K_{nm})$ be points of the Stiefel manifold corresponding in the sense of Proposition~\ref{kraus-as-root} to $\sqrt{C_{\Phi_k}}$ and $\sqrt{C_{\Phi}}$ respectively. The matrices $\sqrt{C_{\Phi_k}}$ and $s_k=(K_1^k, \dots, K_{nm}^k)$, as well as the matrices $\sqrt{C_{\Phi}}$ and $s =(K_1, \dots, K_{nm})$, have the same sets of elements. Then from the convergence of the sequence $\{\sqrt{C_{\Phi_k}}\}$ to $\sqrt{C_{\Phi}}$ it follows that $s_k \to s$ as $k\to \infty$. From the continuity of the map $\pi$ it follows that $h^{-1}(\Phi_k) = \pi(s_k) \to \pi(s) = h^{-1}(\Phi)$ as $k\to \infty$.  This means that $h^{-1}$ is continuous, that is, $h$ is a homeomorphism.

It is known that if $X$ is a topological space and a group $G$ acts on it by homeomorphisms, then the quotient map $\pi\colon X \to X/G$ is an open map (see e.g.~\cite{lee2003introduction}). It implies that the maps $\pi\colon V_n(\CC^{nm^2})\to V_n(\CC^{nm^2})/\mathrm{U}(nm)$ and, hence, $\widetilde{\pi}\colon V_n(\CC^{nm^2})\to \mathcal{C}han(n, m)$ are open.   
\end{proof}

\section{Proof of Theorem~\ref{Stiefel_repr}}
\label{B}

It is known that if a compact Lie group $G$ acts smoothly on a smooth manifold $M$ then the orbit of any point is an embedded smooth submanifold of $M$ (see e.g.~\cite{lee2003introduction}). Since the Lie group $\mathrm{U}(nm)$ is compact and acts smoothly on $V_n(\CC^{nm^2})$, the orbits of its action are embedded smooth submanifolds.

Let $\Phi$ be a quantum channel with Kraus rank $k$, i.e.\ $\operatorname{rank} (C_{\Phi})=k$. Since $C_{\Phi}\geq 0$, we can fix the expansion $C_{\Phi}=\sum_{i=1}^k v_i v_i^{\dagger}$, where $\{v_i\}_{i=1}^k$ are orthonormal eigenvectors of $C_{\Phi}$. This decomposition, as discussed at the beginning of Appendix \ref{A}, corresponds to some fixed Kraus representation with Kraus operators $(K_1, \dots, K_{k}, 0, \dots, 0)$ which represents some point $s$ of Stiefel manifold. If $G=\mathrm{U}(nm)$, the orbit is $Gs \cong G/\Stab s\cong \mathrm{U}(nm)/\mathrm{U}(nm-k)$, because any element of stabilizer must act identically on the first $k$ Kraus operators due to their linear independence. Then the orbit is diffeomorphic to a complex Stiefel manifold $V_k(\CC^{nm})$ (if $k=mn$, this manifold is the unitary group $U(mn)$). Below we give rigorous proof with an explicit diffeomorphism.

Since the vectors $\{v_i\}_{i=1}^k$ are linearly independent, the matrices $\{v_iv_j^{\dagger}\}_{i,j=1}^k$ are also linearly independent. Indeed, this set is a basis in the space of all linear operators on $\operatorname{supp}(C_{\Phi})=\Span(\{v_i\}_{i=1}^k)$. As discussed at the beginning of Appendix \ref{A}, every point in $\pi^{-1}(\Phi)$
uniquely corresponds to a certain decomposition of the Choi matrix in the form
\begin{equation*}
C_{\Phi} = \sum_{l=1}^{nm} w_l w_l^{\dagger},
\end{equation*}
where some $w_l$ are possibly $0$. Since $\sum_{i=1}^k v_i v_i^{\dagger} = \sum_{l=1}^{nm} w_l w_l^{\dagger}$, vectors $w_l$ belong to  $\operatorname{supp}(C_{\Phi})$. Indeed, let some vector $w_q$ have a component perpendicular to $\operatorname{supp}(C_{\Phi})$, i.e.\ $w_q=w_{\parallel}+w_{\perp}$, where $w_{\parallel}\in \operatorname{supp}(C_{\Phi})$, $w_{\perp}\in \Ker(C_{\Phi})$, $w_{\perp}\neq 0$. 
Then, on the one hand,
\begin{equation*}
    w_{\perp}^{\dagger}C_{\Phi}w_{\perp} = \sum_{i=1}^k w_{\perp}^{\dagger}v_i v_i^{\dagger}w_{\perp}=0,
    \end{equation*}
while on the other hand, 
\begin{equation*}
    w_{\perp}^{\dagger}C_{\Phi}w_{\perp} =w_{\perp}^{\dagger} w_q w_q^{\dagger} w_{\perp} + \sum_{l\neq q} w_{\perp}^{\dagger} w_l w_l^{\dagger} w_{\perp} = \|w_{\perp}\|^4 + \sum_{l\neq q} w_{\perp}^{\dagger} w_l w_l^{\dagger} w_{\perp} \geq \|w_{\perp}\|^4 > 0.
\end{equation*}
We get a contradiction. Thus, all vectors $w_l$ belong to $\operatorname{supp}(C_{\Phi})$.
    
Let $w_l=\sum_{i=1}^k \mu_{li} v_i$. This expansion is unique, and the dependence of the coefficients $\mu_{li}$ on the vectors $w_l$ is smooth, since $\mu_{li} = v_i^{\dagger}w_l= \langle v_i, w_l \rangle$. Then 
\begin{equation*}
C_{\Phi} = \sum_{l=1}^{nm} \left(\sum_{i=1}^k \mu_{li} v_i\right) \left(\sum_{j=1}^k \mu_{lj} v_j\right)^{\dagger} = \sum_{l=1}^{nm}\sum_{i,j=1}^{k} \mu_{li}\overline{\mu_{lj}}v_i v_j^{\dagger}.
\end{equation*}
Then from the linear independence of $\{v_iv_j^{\dagger}\}_{i,j=1}^k$ we obtain
\begin{equation*}
\sum_{l=1}^{nm} \mu_{li}\overline{\mu_{lj}}=\delta_{ij}.
\end{equation*}
Thus, the $nm \times k$ matrix $M = (\mu_{li})$ satisfies the condition $M^{\dagger}M=\mathbb{I}_k$. Conversely, if we have such a matrix $M$, we can construct the set of vectors $\{w_l\}$ and consequently the Kraus representation. Therefore, points of the orbit are smoothly parameterized by matrices $M$ of the size $nm\times k$ satisfying $M^{\dagger}M=\mathbb{I}_k$. Thus the orbit $\pi^{-1}(\Phi)$ is an embedded smooth submanifold diffeomorphic to $V_k(\CC^{nm})$.

\section{Convexity}
\label{C}
Let $X$ be a topological space and $f\colon X \to \mathbb{R}$. A point $x_0$ is a point of local minimum (local maximum)   if there exists an open neighbourhood $U$ of $x_0$, such that for all $x\in U$ holds  $f(x)\geq f(x_0)$ ($f(x)\leq f(x_0)$). If for all $x\in X$ holds $f(x)\geq f(x_0)$ ($f(x)\leq f(x_0)$) then $x_0$ is a point of global minimum (global maximum). If $x_0$ is a point of local but not global minimum (maximum) it is a point of suboptimal minimum (maximum).
 
Recall some definitions and statements of convex analysis \cite{roberts1973convex}. Let $V$  be a real vector space. A subset $K\subset V$ is called convex if for every $x_1, x_2 \in K$ and $\lambda\in [0, 1]$ the point $\lambda x_1 + (1-\lambda)x_2$ belongs to $K$. If $K_1$ and $K_2$ are two convex subsets of two vector spaces, a map $f\colon K_1 \to K_2$ is called convex-linear if  
\begin{equation*}
f(\lambda x_1 + (1-\lambda)x_2) = \lambda f(x_1) + (1-\lambda) f(x_1)
\end{equation*}
for any points $x_1, x_2 \in K_1$ and any $\lambda \in [0, 1]$. A function $f\colon K \to \mathbb{R}$ defined on a convex set $K$ is called convex (concave) if 
$$
f(\lambda x_1 + (1-\lambda) x_2) \leq \,(\geq)\, \lambda f(x_1) + (1-\lambda) f(x_2)
$$ 
for any two points $x_1, x_2 \in K$ and  any $\lambda \in [0, 1]$. Similarly, $f\colon K \to \mathbb{R}$ is called strictly convex (strictly concave) if 
$$
f(\lambda x_1 + (1-\lambda) x_2) < \,(>)\, \lambda f(x_1) + (1-\lambda) f(x_2)
$$
for any two distinct points $x_1, x_2 \in K$ and $0 < \lambda < 1$. A point $x$ of a convex set $K$ is called an extreme point if there are no such distinct $x_1, x_2 \in K$ and $\lambda\in (0, 1)$ that $x=\lambda x_1 + (1-\lambda) x_2$. Let $f\colon K \to \RR$ be a convex (concave) function defined on the convex subset $K$ of a real normed vector space. Then any local minimum (maximum) of the function $f$ is a global minimum (maximum). If $f$ is a strictly convex (concave) function on $K$, then the point where the global minimum (maximum) of $f$ is attained is unique if it exists. As well, any local maximum (minimum) point of $f$ is an extreme point of~$K$.

\printbibliography

\end{document}